\documentclass[letterpaper,11pt]{article}
\usepackage[utf8]{inputenc}
\usepackage{times,calc,parskip,graphicx,natbib,xspace} 
\usepackage{doi} % for DOI in bibtex
\usepackage{hyperref}
\usepackage{amssymb,amsmath,amsthm,amscd,cite,setspace}
\usepackage[margin=1in]{geometry}
\usepackage{complexity}
\usepackage{algpseudocode,algorithm}
\newtheorem{thm}{Theorem}

\newtheorem{lemma}{Lemma}%[chapter]

\newcommand{\perm}{\mathop{\textup{Per}}}
\newcommand{\EX}{\mathbb{E}}
\newcommand{\bigo}[1]{\mathcal{O}(#1)}
\newcommand{\multichoose}[2]{(^{#1+#2-1}_{\hphantom{#1+}#2})}
\newcommand{\shortbinom}[2]{(^{#1}_{#2})}
\newcommand{\newtilde}[1]{\mathbf{#1}}
\newcommand*\Let[2]{\State #1 $\gets$ #2}
\newcommand{\Bkl}{B^{\diamond}_{k,\ell}}
\newcommand{\Bk}{B^{\diamond}_{k}}

\def\BS/{Boson Sampling}
\title{The Classical Complexity of \BS/}
\author{Peter Clifford\\
 {\normalsize Department of Statistics}\\
  {\normalsize University of Oxford}\\
  {\normalsize United Kingdom}\\
  \and
  Rapha\"el Clifford\\
  {\normalsize Department of Computer Science}\\
  {\normalsize University of Bristol}\\
  {\normalsize  United Kingdom}\\
  }

\date{\vspace{-1ex}}  

\begin{document}
\maketitle
\thispagestyle{empty}

\begin{abstract}
We study the classical complexity of the exact \BS/ problem where the objective is to produce provably correct random samples from a particular quantum mechanical distribution.  The computational framework was proposed in STOC '11 by \citeauthor{AA:2011} in 2011 as an attainable demonstration of `quantum supremacy', that is a practical quantum computing experiment able to produce output at a speed beyond the reach of classical (that is non-quantum) computer hardware.  Since its introduction \BS/ has been the subject of intense international research in the world of quantum computing.  On the face of it, the problem is challenging for classical computation.  \citeauthor{AA:2011} show that exact \BS/ is not efficiently solvable by a classical computer unless $\P^{\#\P}  = \BPP^{\NP}$ and the polynomial hierarchy collapses to the third level. 

The fastest known exact classical algorithm for the standard \BS/ problem  requires $\bigo{\multichoose{m}{n}\, n 2^n }$ time to produce samples for a system with input size $n$ and $m$ output modes, making it infeasible for anything but the smallest values of $n$ and $m$.  We give an algorithm that is much faster, running in $\bigo{n 2^n + \poly(m,n)}$ time and $\bigo{m}$ additional space. The algorithm is simple to implement and has low constant factor overheads. As a consequence our classical algorithm is able to solve the exact \BS/ problem for system sizes far beyond current photonic quantum computing experimentation, thereby significantly reducing the likelihood of achieving near-term quantum supremacy in the
context of \BS/. 

\end{abstract}

%\newpage
%\setcounter{page}{1}
\section{Introduction}

The promise of significantly faster quantum algorithms for problems of practical interest is one of the most exciting prospects in computer science. Most famously, a quantum computer would allow us to factorise integers in polynomial time \citep{Shor:1997} and perform unstructured search on an input of $n$ elements in $\bigo{\sqrt{n}}$ time \citep{Grover:1996}. In the short term however, despite impressive progress in recent years, there are still considerable challenges to overcome before we can build a fully universal quantum computer.   Until this time, the question of how to design quantum experiments which fall short of fully universal computation but which still show  significant quantum speed-up over known classical (that is non-quantum) algorithms remains of great interest. In a breakthrough contribution in STOC '11 \citet{AA:2011} proposed an experimental set-up in linear optics known as \BS/ as a potentially practical way to do just this.    

Since its introduction, \BS/ has attracted a great deal of attention with numerous experimental efforts around the world attempting implementations for various problem sizes~\citep[see e.g.\@][]{bentivegna2015experimental,spring2013boson,broome2013photonic,tillmann2013experimental,crespi2013integrated,spagnolo2014experimental,LSS:2016}. The ultimate goal is to exhibit a physical quantum experiment of such a scale that it would be hard if not impossible to simulate the output classically and thereby to establish so-called `quantum supremacy'.    In terms of the physical \BS/ experiment, $n$ corresponds to the number of photons and $m$ the number of output modes and increasing either of these is difficult in practice. Progress has therefore been slow \citep[see][and the references therein]{LBR:2017} with the current experimental record being $n = 5,  m = 9$ \citep[]{Wang2016}.  

Translated into conventional computing terms, the task is to generate independent random samples from a specific probability distribution on multisets of size $n$ with elements in the range $1$ through $m$.  For \BS/, the  probability of each multiset is related to the permanent of an $n \times n$ matrix built from possibly repeated rows of a larger $m \times n$ matrix where the multiset determines the rows used in the construction. The  \textit{\BS/ problem} is to produce such samples either with a quantum photonic device or with classical computing hardware. 

On the face of it, \BS/ is a challenging problem for classical computation.  \citet{AA:2011} showed that exact \BS/ is not solvable in polynomial time by a classical computer unless $\P^{\#\P}  = \BPP^{\NP}$ and the polynomial hierarchy collapses to the third level. 

Without a clear understanding of the classical complexity of the problem it is difficult to determine the minimum experimental size needed to establish quantum supremacy. In the absence of a fast classical algorithm for \BS/, speculation about the minimum size of $n$ (with $m=n^2$) has varied a great deal in the literature from e.g.\@ $7$ \citep[]{LSS:2016}, between $20$ and $30$ \citep[]{AA:2014} and all the way up to $50$ \citep[]{LBR:2017}.   Broadly speaking the lowest estimate corresponds to the limits of the previous fastest classical \BS/ algorithm and the highest with an assumption that the polynomial hierarchy will not collapse.  

Brute force evaluation of the probabilities of each multiset as a preliminary to random sampling, requires the calculation of $\multichoose{m}{n}$  permanents of $n \times n$ matrices, each one of which takes $\bigo{n 2^n}$ time with the fastest known algorithms.  If $m \geqslant n^2$, as considered by \citet{AA:2011}, the total running time is $\Theta(2^n \mathrm{e}^n (m/n)^n  n^{1/2})$. The space usage of such a naive approach can be reduced to the  storage needed for the required sample size  \citep[see e.g.\@][]{Efraimidis:2015}, however the running time for anything but the smallest values of $n$ and $m$ remains prohibitive. Previous to the work we present here, no faster provably correct classical \BS/ algorithm was known. 

Recently \citet{NSCJML:2017}  gave strong numerical evidence suggesting that a specially designed Markov chain Monte Carlo (MCMC) implementation can sample approximately from the collision-free version of \BS/.    The time cost per sample for their MCMC method for problems with size $n \leqslant 30$ corresponds to that of computing around two hundred $n \times n$ permanents, making it computationally feasible for the problem sizes they consider.  This was the first practical evidence that classical \BS/ might be possible faster than the naive brute force approach and is a significant motivation for the work we present here.  
%Furthermore, in his recent blog \citet{Aaronson:blog} reports that he has long suspected the classical complexity of approximate sampling would be dominated by a factor of $2^n$. Possible contributions from terms involving $m$ are not discussed, however.  We provide substance to Aaronson's insight and, for the more challenging problem of exact sampling, we show precisely how such terms interact in Theorem \ref{thm2n}. 

We show that \BS/ on a classical computer can be performed both \emph{exactly} and dramatically faster than the naive algorithm, significantly raising the bar for the minimum size of quantum experiment needed to establish quantum supremacy.   Our sampling algorithm also provides the probability associated with any given sample at no extra computational cost. 
\vspace{0.5em}
\begin{thm}\label{thm2n} The time complexity of  \BS/ is:
\[ \bigo{n 2^n +\poly(m,n)}, \]
where $\poly(m,n) = \bigo{m n^2}.$ The additional space complexity on top of that needed to store the input is $\bigo{m}$.
\end{thm}

A particularly attractive property of the running time we give is that the exponentially growing term depends only on $n$, the size of the multisets we wish to sample, and not $m$ which is potentially much larger.   In the most natural parameter regime when $m$ is polynomial in $n$ our algorithm is therefore approximately $\multichoose{m}{n}$ times faster than the fastest previous method. Our algorithm is both straightforward to implement and has small constant factor overheads, with the total computational cost to take one sample  approximately equivalent to that of calculating two $n \times n$ permanents.    This gives a remarkable resolution to the previously open question of the relationship between the classical complexity of computing the  permanent of an $n \times n$ matrix and \BS/. For an implementation of the algorithm, see \citet{Boson:Rpackage}.

Based on the time complexity we demonstrate in Theorem \ref{thm2n} and bearing in mind recent feasibility studies for the calculation of large permanents \citep{wu2016computing} we find support for $n = 50$ as the threshold for quantum supremacy for the exact \BS/ problem.

Our sampling algorithm also applies directly to the closely related problem of scattershot \BS/ \citep[]{bentivegna2015experimental}. From a mathematical perspective this produces no additional complications beyond the specification of a different $m \times n$ matrix. Once the matrix is specified we can apply our new sampling algorithm directly. 

Our techniques draw heavily on the highly developed methods of hierarchical sampling commonly used in Bayesian computation~\citep[see e.g.\@][]{liu2008monte,GLPR:2015}; specially the use of auxiliary variables and conditional simulation. We believe this cross-fertilisation of ideas is novel and opens up the development of classical algorithms for sampling problems in quantum computation. It raises the intriguing  question of whether similar techniques can be applied to get significant improvements in time and space for other proposed quantum supremacy projects such as, for example, the problem of classically simulating quantum circuits \citep{BISBDJMH:2016}.

The algorithmic speed-up we achieve relies on a number of key innovations. As a first step we expand the sample space of multisets into a product set $[m]^n$,  i.e.\ the space of all arrays $\newtilde{r} = (r_1,\dots,r_n)$, with each element in the range $1$ to $m$. This simple transformation allows us to expose combinatorial properties of the sampling problem which were otherwise opaque.  We then develop explicit expressions for the probabilities of subsequences of arrays in this set. Application of the chain rule for conditional probabilities then leads to our first speed-up with an $\bigo{mn3^n}$ time algorithm for \BS/. For our main result in Theorem \ref{thm2n}, we introduce two further innovations. First we exploit the Laplace expansion of a permanent to give an efficient amortised algorithm for the computation of the permanent of a number of closely related matrices. We then expand the sample space further with an auxiliary array representing column indices and show how a conditioning argument allows \BS/ to be viewed as sampling from a succession of coupled conditional distributions. The final result is our improved time complexity of $\bigo{n2^n  + mn^2}$. 

A consequence of our analysis is that we are also able to provide an efficiently computable bound on the probability that a sampled multiset drawn from the \BS/ distribution will contain duplicated elements.  This enables us to bound the performance of rejection sampling variants of our sampling algorithm when applied to the  so-called `collision-free' setting described by \citeauthor{AA:2011}.

One further application of our new sampling algorithm lies in tackling the question of whether the output of physical \BS/ devices correspond to theoretical predictions.    There has been considerable interest in this area \citep{AA:2014,tichy2014stringent,wang2016certification,walschaers2016statistical} however previously developed methods have had to work under the assumption that we cannot sample from the \BS/ distribution for even moderate sizes of $n$ and $m$.  Armed with our new sampling algorithm, we can now sample directly from the \BS/ distribution for a much larger range of parameters, allowing a far greater range of statistical tests to be applied.

\subsection*{Related work and computation of the permanent}

\citet{TD:2004} were the first to recognise that studying the complexity of sampling from low-depth quantum circuits could be useful for demonstrating a separation between quantum and classical computation. Since that time the goal of finding complexity separations for quantum sampling problems has received a great deal of interest. We refer the interested reader to \citet{LBR:2017} and \citet{HM:2017} for recent surveys on the topic.

There is an extensive literature on the problems of sampling from unwieldy distributions with applications in all the main scientific disciplines. Much of this work has relied on MCMC methods where the target distribution can be identified as the equilibrium distribution of a Markov chain \citep[see e.g.\@][]{hastings1970monte}.  For \BS/ the most closely related work to ours is that of \citet{NSCJML:2017} mentioned earlier.  In their paper the authors also set out the hurdles that need to be overcome for a physical \BS/ experiment to reach higher values of $n$.  Their MCMC approach does not however permit them to estimate how well their sampling method approximates the true \BS/ distribution, other than through numerical experimentation. Our new exact sampling algorithm is therefore the first provably correct algorithmic speed-up for \BS/. In fact our algorithm also has smaller constant factor overheads, costing approximately two permanent calculations of $n \times n$ matrices for all values of $n$ (assuming $m$ remains polynomial in $n$). 

There are MCMC algorithms which permit exact sampling,  for example via `coupling from the past', where guarantees of performance can be provided \citep[]{propp-wilson:exact-sampling}. However, the range of problems for which this approach has proven applicable is necessarily more limited and there has yet to be a successful demonstration for the \BS/ problem.  
 
\paragraph{Permanents:} Computation of the permanent of a matrix was shown to be  \#\P-hard by \citet{Valiant:1979} and hence is unlikely to have a polynomial time solution.  Previously \citet{Ryser:1963}  had proposed an $O(n^2 2^n)$ time algorithm to compute the permanent of an $n \times n$ matrix.  This was sped up to $O(n 2^n)$ time by~\citet{NW:1978} and many years later a related formula with the same time complexity was given by~\citet{Glynn:2010}.  It is  Glynn's computational formula for the permanent which we will use as the basis of our sampling algorithms. 

Let $M = (m_{i,j})$ be an $n \times n$ matrix with $m_{i,j} \in \mathbb{C}$, and let $\perm M$ denote its permanent, then 
\begin{equation}\label{eq:glynn}
\perm M = \frac{1}{2^{n-1}} \sum_{\delta} \left( \prod_{k = 1}^n \delta_k \right) \prod_{j=1}^n \sum_{i=1}^n \delta_i m_{i,j} 
\end{equation}

where $\delta  \in \{-1,1\}^{n}$ with $\delta_1 =1$.

A naive implementation of this formula would require $\bigo{n^2 2^n}$ time however by considering the $\delta$ arrays in Gray code order this is reduced to $O(n 2^n)$ time. Note that further reduction in the time complexity may be possible when $M$ has repeated rows or columns; see \citet[][Appendix B]{Tichy:2011} and \citet[][Appendix D]{shchesnovich2013asymptotic}.  In the special case of computation of the permanent over $\poly(n)$ bit integers, \citet{bjorklund2016below} has a faster $\bigo{n2^{n-\Omega(\sqrt{n/\log{n}})}}$ time algorithm for computing the permanent.

\section{Problem specification}\label{sec:specification}
In  mathematical terms the \BS/ problem is as follows.

Let $m$ and $n$ be positive integers and consider all multisets of size $n$ with elements in $[m]$, where $[m] = \{1,\dots,m\}$.  Each multiset can be represented by an array  $\newtilde{z} = (z_1,\dots,z_n)$ consisting of elements of the multiset in non-decreasing order: $z_1 \leqslant z_2 \leqslant \dots \leqslant z_n$. 
We denote the set of distinct values of $\newtilde{z}$ by $\Phi_{m,n}$. The cardinality of $\Phi_{m,n}$ is known to be $\multichoose{m}{n}$ -- see \citet{Feller:1968}, for example. We define $\mu(\newtilde{z}) = \prod_{j=1}^m s_j!$ where $s_j$ is the multiplicity of the value $j$ in $\newtilde{z}$.

Now let $A = (a_{ij})$ be a complex valued $m \times n$ matrix consisting of the first $n$ columns of a given $m \times m$ Haar random unitary matrix. For each $\newtilde{z}$, build an $n \times n$ matrix $A_{\newtilde{z}}$ where the $k$-th row  of $A_{\newtilde{z}}$ is row $z_k$ in $A$ for $k=1,\dots,n$ and define a probability mass function (pmf) on $\Phi_{m,n}$ as

\begin{equation}\label{eq:Bz_s}
q(\newtilde{z}) 
= \frac{1}{\mu(\newtilde{z})} \left| \perm A_{\newtilde{z}}\right|^2 
\overset{\text{defn}}{=} \frac{1}{\mu(\newtilde{z})} \left| \sum_{\sigma} \prod_{k=1}^n a_{z_k \sigma_k}\right|^2, \quad \newtilde{z} \in \Phi_{m,n},
\end{equation}
where $\perm A_{\newtilde{z}}$ is the permanent of $A_{\newtilde{z}}$ and in the definition the summation is for all  $\sigma \in \pi[n]$,  the set of permutations of $[n]$.  

The Cauchy-Binet formula for permanents \citep[see][page 579]{marcus1965permanents} can be used to show the function $q$ is indeed a pmf.  There is a demonstration from first principles in the proof of Lemma \ref{th:subsequence}. See also \citet[Theorem 3.10]{AA:2011} for a demonstration from physical principles. 

The computing task is to simulate random samples from the pmf $q(\newtilde{z})$.
\section{Exact \BS/ Algorithms}

We give two algorithms for \BS/.  Algorithm A runs in $\bigo{mn3^n}$ time to create a single sample. This is already a significant speed-up over the  $\Theta{\left(\multichoose{m}{n} n 2^n\right)}$ time complexity of the fastest solution in the literature. Algorithm A prepares the way for Algorithm B our main result with running time specified in Theorem \ref{thm2n}.  
 
We will repeatedly need to sample from an unnormalised discrete probability mass function (weight function).  To establish the complexity of our \BS/ algorithms we need only employ the most naive linear time method  although faster methods do exist. In particular, if the probability mass function is described by an array of length $m$ then it is possible to sample in $\bigo{1}$ time after $\bigo{m}$ preprocessing time, using Walker's alias method~\citep{Walker:1974,KP:1979}. 
 
\paragraph{Expanding the sample space:} We start by expanding the sample space $\Phi_{m,n}$ and consider a distribution on the product space $[m]^n$, i.e.\ the space of all arrays $\newtilde{r} = (r_1,\dots,r_n)$, with each element in the range $1$ to $m$.  Since the permanent of a matrix is invariant to row reordering, we have $\perm A_{\newtilde{r}} = \perm A_{\newtilde{z}}$ when $\newtilde{r}$ is any permutation of $\newtilde{z}$.  Furthermore the number of distinguishable rearrangements of $\newtilde{z}$ is $n!/\mu(\newtilde{z})$.     

We claim that in order to sample from $q(\newtilde{z})$ we can equivalently sample from the pmf
\begin{equation} \label{def:pr}
p(\newtilde{r}) 
= \frac{1}{n!} \left| \perm A_{\newtilde{r}} \right|^2 = \frac{1}{n!} \left| \sum_{\sigma} \prod_{i=1}^n a_{r_i \sigma_i}\right|^2, \quad \newtilde{r} \in [m]^n.
\end{equation}
In other words sample $\newtilde{r}$ from $p(\newtilde{r})$  and then rearrange the elements of $\newtilde{r}$ in non-decreasing order to give $\newtilde{z}$.

This follows since for any $\newtilde{z}$ there are $n!/\mu(\newtilde{z})$ equally likely values of $\newtilde{r}$ in the expanded sample space, i.e.\ $p(\newtilde{r}) = p(\newtilde{z})$ for all such values of $\newtilde{r}$. So by the addition rule for probabilities
\[
q(\newtilde{z}) = \frac{n!}{\mu(\newtilde{z})} p(\newtilde{z}) = \frac{n!}{\mu(\newtilde{z})} \frac{1}{n!} \left| \perm A_{\newtilde{z}}\right|^2 = \frac{1}{\mu(\newtilde{z})} \left| \perm A_{\newtilde{z}}\right|^2,\quad \newtilde{z} \in \Phi_{m,n},
\]
as claimed.

This straightforward reformulation will make the task of efficient \BS/ significantly simpler. 
Our approach focuses on the pmf $p(\newtilde{r})$.   

We begin by deriving the marginal pmfs of the leading subsequences
 of $(r_1,\dots,r_n)$ in the following lemma. This will then provide the pmfs of successive elements of
 $(r_1,\dots,r_n)$ conditional on previous values. In that way
 we are able to rewrite $p(r_1,\dots,r_n)$ as a product of conditional pmfs,
 using the chain rule of probability, and exploit the chain for progressively
 simulating $r_1,\dots,r_n$.

\vspace{0.5em}
\begin{lemma}\label{th:subsequence}
The joint pmf of the subsequence $(r_1,\dots,r_k)$  is given by 
\[
p(r_1,\dots,r_k) = \frac{(n-k)!}{n!}\sum_{c \in \mathcal{C}_k} \left|\perm A_{r_1,\dots,r_k}^c\right|^2,\quad k = 1, \dots, n,
\]
where $\mathcal{C}_k$ is the set of $k$-combinations taken without replacement from $[n]$ and $A_{r_1,\dots,r_k}^c$ is the matrix formed from rows $(r_1,\dots,r_k)$ of the columns $c$ of $A$.
\end{lemma}
\begin{proof}
The convention will be that $c$ is a set and $(r_1,\dots,r_k)$ is an array with possibly multiple instances of its elements. The invariance of permanents under column and row permutations removes any ambiguity in the interpretation of $\perm A^c_{r_1,\dots,r_k}$.
  
For completeness we start by confirming that $p(\newtilde{r})$ is a pmf, in other words $\sum_{\newtilde{r}} p(\newtilde{r}) = 1.$
From the definition of $p(\newtilde{r})$ and multiplying by $n!$, we have 
\begin{align}
n! \sum_{\newtilde{r}} p(\newtilde{r}) 
&= \sum_{\newtilde{r}} \left|\sum_\sigma \prod_{i=1}^n a_{r_i \sigma_i} \right|^2, \sigma \in \pi[n]\nonumber \\
&= \sum_{\newtilde{r}} \left(\sum_\sigma \prod_{i=1}^n a_{r_i \sigma_i }\right)\left(\sum_\tau \prod_{i=1}^n \bar{a}_{r_i \tau_i}\right),\tau \in \pi[n] \nonumber\\
&= \sum_\sigma \sum_\tau  \left(\sum_{\newtilde{r}} \prod_{i=1}^n a_{r_i \sigma_i} \bar{a}_{r_i \tau_i} \right)\nonumber\\
&= \sum_\sigma \sum_\tau \prod_{i=1}^n \sum_{k=1}^m a_{k,\sigma_i} \bar{a}_{k \tau_i}.\nonumber 
\end{align}
The product is $0$ when $\sigma_i \neq \tau_i$ for any $i$, by the orthonormal property of $A$.  Otherwise when $\sigma = \tau$ the product is $1$ so the final expression reduces to $n!$ and $\sum_{\newtilde{r}} p(\newtilde{r}) = 1$ as expected. 

Now looking at the case $k=1$ we have similarly
\begin{align}\label{eqB}
n! p(r_1) &= \sum_{r_2\dots r_n} \left(\sum_\sigma \prod_{i=1}^n a_{r_i \sigma_i }\right)\left(\sum_\tau \prod_{i=1}^n \bar{a}_{r_i \tau_i}\right)\nonumber\\
&= \sum_\sigma \sum_\tau a_{r_1 \sigma_1} \bar{a}_{r_1 \tau_1 } \prod_{i=2}^n \sum_{k=1}^m a_{k,\sigma_i} \bar{a}_{k \tau_i}. \nonumber
\end{align}

Again for this to be non-zero, each of the terms in the product must be $1$.  This means $\sigma_i = \tau_i, i=2, \dots, n$ which implies $\sigma = \tau$.

It follows that
\[n!p(r_1) = \sum_\sigma |a_{r_1 \sigma_1}|^2 = (n-1)! \sum_{k=1}^n |a_{r_1,k}|^2,\]
as claimed,  so that 
\begin{equation}\label{eqC}p(r_1) = \frac{1}{n}\left[|a_{r_1,1}|^2 + \dots, + |a_{r_1,n}|^2\right],\; r_1 \in [m],\nonumber \end{equation}
the sum of the squared moduli of elements of row $r_1$ divided by $n$. 

In the same way,
\[ n! p(r_1,r_2) =  \sum_\sigma \sum_\tau a_{r_1 \sigma_1} \bar{a}_{r_1 \tau_1 } a_{r_2 \sigma_2} \bar{a}_{r_2 \tau_2 } \left(\prod_{i=3}^n \sum_{k=1}^m a_{k,\sigma_i} \bar{a}_{k \tau_i}\right).
\] 
Again the only non-zero case is when $\sigma_i = \tau_i, i= 3,\dots,n$. However it does not imply that $\sigma = \tau$, since we can also have $\sigma_1 = \tau_2, \sigma_2 = \tau_1$.

Consequently,
\begin{align}\label{eqD}
 n!p(r_1, r_2) &= \mathop{\sum_\sigma \sum_\tau}_{\sigma_3,\dots,\sigma_n = \tau_3,\dots,\tau_n}  a_{r_1 \sigma_1} \bar{a}_{r_1 \tau_1 } a_{r_2 \sigma_2} \bar{a}_{r_2 \tau_2 }\nonumber\\
&= (n-2)! \sum_{c \in\mathcal{C}_2} \left(\sum_{\sigma \in \pi(c)} a_{r_1 \sigma_1} \bar{a}_{r_2 \sigma_2}\right)^2\nonumber\\
& = (n-2)! \sum_{c \in\mathcal{C}_2} \left|\perm A_{r_1,r_2}^c\right|^2,\nonumber
\end{align}
where $\mathcal{C}_2$ is the set of $2$-combinations (without replacement) from $[n]$, $\pi(c)$ is the set of permutations of $c$ and $A_{r_1,r_2}^c$ is a matrix formed from rows $(r_1,r_2)$ and columns $c$  of $A$.

In the same way
\[n!p(r_1,\dots,r_k) = (n-k)! \sum_{c \in\mathcal{C}_k} \left|\perm A_{r_1,\dots,r_k}^c\right|^2,\]
so that
\begin{equation}\label{eqE}
p(r_1,\dots,r_k) = \frac{(n-k)!}{n!}\sum_{c \in\mathcal{C}_k} \left|\perm A_{r_1,\dots,r_k}^c\right|^2,\quad k = 1, \dots, n.\nonumber
\end{equation}
\end{proof}

Our first \BS/ algorithm which we term Algorithm A samples from $p(\newtilde{r})$ by expressing this as a chain of conditional pmfs,  
\[p(\newtilde{r}) = p(r_1) p(r_2|r_1) p(r_3|r_1,r_2) \dots p(r_n | r_1,r_2,\dots,r_{n-1}).\]
Sample $r_1$ from $p(r_1)$, $r_1 \in [m]$. Subsequently for stages $k = 2,\dots,n$, sample $r_k$ from the conditional pmf $p(r_k|r_1,\dots,r_{k-1})$, $r_k \in [m]$ with $(r_1,\dots,r_{k-1})$ fixed.  At completion, the array $(r_1,\dots,r_n)$ at stage $n$ will have been sampled from the pmf $p(\newtilde{r})$. Sorting $(r_1,\dots,r_n)$ in non-decreasing order, the resulting multiset $\newtilde{z}$ is sampled from the \BS/ distribution $q(\newtilde{z})$.

\begin{algorithm}
  \caption{Boson Sampler: single sample $\newtilde{z}$ from q($\newtilde{z}$) in $\bigo{mn3^n}$ time
    \label{alg:basic}}
  \begin{algorithmic}[1]
    \Require{$m$ and $n$ positive integers;  $A$ first $n$ columns of $m \times m$ Haar random unitary matrix}
    \State $\newtilde{r} \gets \varnothing$ \Comment{\sc Empty array}
    \For{$k \gets 1 \textrm{ to } n$}
    \State $w_i \gets \sum_{c \in\mathcal{C}_k} |\perm A_{(\newtilde{r},i)}^c|^2, i \in [m]$ \Comment{Make indexed weight array $w$}
    \State $x \gets \textrm{Sample}(w)$ \Comment{Sample index $x$ from $w$}
    \State $\newtilde{r} \gets (\newtilde{r}, x)$ \Comment{Append $x$ to $\newtilde{r}$}
    \EndFor
    \State $\newtilde{z} \gets \textrm{IncSort}(\newtilde{r})$ \Comment{Sort $\newtilde{r}$ in non-decreasing order}
    \State \Return{$\newtilde{z}$}
  \end{algorithmic}
\end{algorithm}
\vspace{0.5em}
\begin{thm}\label{thm3n}
Algorithm A samples from the \BS/ distribution with time complexity $\bigo{m n 3^n}$ per sample. The additional space complexity of Algorithm A on top of that needed to store the input is $\bigo{m}$ words.
\end{thm}
\begin{proof}
The correctness follows from the chain rule for conditional probabilities and so we now analyse the running time. From Lemma \ref{th:subsequence}, evaluation of the pmf $p(r_1)$ for $r_1 \in [m]$ involves the sum of squared moduli in each row of $A$, a total of $\bigo{mn}$ operations. Sampling from this pmf takes $\bigo{m}$ time.  More generally we can write $p(r_k|r_1,\dots,r_{k-1}) = p(r_1,\dots,r_k)/p(r_1,\dots,r_{k-1})$.  Note that $r_k$ does not appear in the denominator, so to sample $r_k$ from this conditional pmf we can equivalently sample from the pmf that is proportional to the numerator, considered purely as function of $r_k$ with $r_1,\dots,r_{k-1}$ fixed. The numerator involves the calculating of several $k \times k$ permanents.  With the best known algorithms for computing the permanent, each requires  $\bigo{k 2^k}$ operations. There are $\shortbinom{n}{k}$ terms in $\mathcal{C}_k$ for each value of $k$, making a combined operation count at stage $k$ of $\bigo{m \shortbinom{n}{k} k 2^k}$.  Using 
\[\sum_{k=1}^n k 2^k \binom{n}{k} = \frac23 n 3^n,\] the total operation count is then $\bigo{m n3^n}$. We only need to store the result of one permanent calculation at a time and also one list of $m$ probabilities at a time.  This gives a total space usage of $\bigo{m}$ excluding the space needed to store the original matrix $A$. 
\end{proof}

We now prepare to prove Theorem~\ref{thm2n}.  
\paragraph{Laplace expansion:}
We make extensive use of the Laplace expansion for permanents \citep[see][page 578]{marcus1965permanents}, namely that for any $k \times k$ matrix $B = (b_{i,j})$,
\begin{equation}\label{eq:Laplace}
\perm B = \sum_{\ell = 1}^k b_{k,\ell} \perm \Bkl, \nonumber
\end{equation}
where $\Bkl$ is the submatrix with row $k$ and column $\ell$ removed.  Note that $\Bkl$ only depends on $\Bk$ the submatrix of $B$ with the $k$-th row removed.  An important consequence is that when $B$ is modified by changing its  $k$-th row, the modified permanent can be calculated in $\bigo{k}$ steps,  provided the values $\{\perm \Bkl\}$ are available. 

We show that computation of all of the values $\{\perm \Bkl, \ell \in [k]\}$ has the same time complexity as computing $\perm B$, the permanent of a single $k \times k$ matrix.  We will later take advantage of this fast amortised algorithm in the proof of Theorem \ref{thm2n} to quickly compute a set of permanents of matrices each with one row differing from the other.   
\vspace{0.5em}
\begin{lemma}\label{lemma1}
Let $B$ be a $k \times k$ complex matrix and let $\{\Bkl\}$ be submatrices of $B$ with row $k$ and column $\ell$ removed, $\ell \in [k]$. The collection $\{\perm  \Bkl, \ell \in [k]\}$ can be evaluated jointly in $\bigo{k 2^k}$ time and $\bigo{k}$ additional space.  
\end{lemma}
\begin{proof}[Proof of Lemma]
By applying Glynn's formula~\eqref{eq:glynn} for a given value of $\ell$ we have:
\[
\perm  \Bkl = \frac{1}{2^{k-2}}\sum_{\delta}  \left( \prod_{i = 1}^{k-1} \delta_i \right) \prod_{j \in [k] \setminus \ell}  v_j(\delta), \quad \ell \in [k],
\]
where $\delta  \in \{-1,1\}^{k-1}$ with $\delta_1 =1$ and  $ v_j(\delta) = \sum_{i=1}^{k-1} \delta_i b_{ij}$. 

Exploiting the usual trick of working through values of $\delta$ in Gray code order, the terms $\{v_j(\delta), j \in [k]\}$ can be evaluated in $\bigo{k}$ combined time for every new $\delta$. This is because successive $\delta$ arrays will differ in one element. The product of the $v_j(\delta)$ terms can therefore also be computed in $\bigo{k}$ time giving $\bigo{k2^k}$ time to compute $\perm  \Bkl$ for a single value of $\ell$, but of course this has to be replicated $k$ times to cover all values of $\ell$. 

To compute $\{\perm  \Bkl  \ell \in [k]\}$ more efficiently we observe that each product $\prod_{j \in [k] \setminus \ell}  v_j(\delta)$ can be expressed as $f_\ell b_\ell$  where $f_\ell = \prod_{j=1}^{\ell-1}v_j(\delta), \ell = 2,\dots,k$ and $b_\ell = \prod_{j=\ell+1}^k v_j(\delta), \ell = 1,\dots,k-1$ are forward and backward cumulative products, with $f_1 = b_k = 1$.  
We can therefore compute all of the partial products $\prod_{j \in [k] \setminus \ell}  v_j(\delta)$ in $\bigo{k}$ time, giving an overall total time complexity of $\bigo{k2^k}$ for jointly computing $\{\perm  \Bkl, \ell \in [k]\}$. Other than the original matrix, space used is dominated by the array of $k$ accumulating partial products. 
\end{proof}  

%With only small modifications, the fast amortised algorithm of Lemma~\ref{lemma1} can be made to run in the stated time complexity for any matrix $A$, removing the requirement that it be Haar random completely~\citep[see e.g.\@][]{Boson:Rpackage}. 

Furthermore the computation time has constant factor overheads similar to that of computing $\perm B$ (see Appendix).

\paragraph{Auxiliary column indices:}
We now expand the sample space even further with an auxiliary array $\boldsymbol{\alpha}= (\alpha_1,\dots,\alpha_n)$, where $\boldsymbol{\alpha} \in \pi[n]$, the set of permutations of $[n]$. 

As with Lemma 1, the purpose is to create a succession of pmfs for leading subsequences of the array $(r_1,\dots,r_n)$ and thereby provide successive conditional pmfs so that  $r_1,\dots,r_n$ can be simulated progressively. The novelty here is to introduce an overall conditioning variable $\boldsymbol{\alpha}$ so that sampling can be carried out more rapidly under that condition while demonstrating that the marginal distribution still correctly provides the pmf $p(\newtilde{r})$ in \eqref{def:pr}.  

Define 
\[
\phi(r_1,\dots,r_{k}|\boldsymbol{\alpha}) =  \frac{1}{k!} \left| \perm A_{r_1,\dots,r_k}^{[n] \setminus \{\alpha_{k+1},\dots,\alpha_n\}}\right|^2,  \quad k = 1,\dots,n-1.
\]
For notational convenience let $e_k = \phi(r_1,\dots,r_{k}|\boldsymbol{\alpha})$ and $d_k = \sum_{r_k} \phi(r_1,\dots,r_{k}|\boldsymbol{\alpha})$ for $k = 1,\dots,n-1$ with $e_n  = p(r_1,\dots,r_{n})$ and $d_n = p(r_1,\dots,r_{n-1})$. Note that 
$\phi(r_1,\dots,r_{k}|\boldsymbol{\alpha})$ is a pmf on $[m]^k$, equivalent to \eqref{def:pr} using the subset  $\{\alpha_1,\dots,\alpha_k\}$ of columns of $A$, so that in particular $d_1 = 1$. 
 
\vspace{0.5em} 
\begin{lemma}\label{thm:phi}
With the preceding notation, let $\phi(\newtilde{r}|\boldsymbol{\alpha}) = \prod_{k=1}^n e_k/d_k$
then $p(\newtilde{r}) =  \EX_{\boldsymbol{\alpha}} \{\phi(\newtilde{r}|\boldsymbol{\alpha})\}$ where the expectation is taken over $\boldsymbol{\alpha}$, uniformly distributed on $\pi[n]$ for fixed $\newtilde{r}$.
\end{lemma}
\begin{proof} For the particular case $k = n-1$, Lemma \ref{th:subsequence} shows that $p(r_1,\dots,r_{n-1})$ can be written as a mixture of component pmfs and expressed as an expectation over $\alpha_n$, i.e.\
\[
p(r_1,\dots,r_{n-1}) = \frac{1}{n} \sum_{j=1}^n \frac{1}{(n-1)!}\left|\perm A^{[n] \setminus \, j}_{r_1,\dots,r_{n-1}}\right|^2 
= \EX_{\alpha_n} \{\phi(r_1,\dots,r_{n-1}|\boldsymbol{\alpha})\},
\]
where $\alpha_n$ is uniformly distributed on $[n]$. In other words $d_n$ is the expectation of $e_{n-1}$ over $\alpha_n$. More generally when the leading subsequence is $(r_1,\dots,r_k)$ and the active columns of $A$ are $[n] \setminus \{\alpha_{k+1},\dots,\alpha_n\}$, we have

\[
\sum_{r_k} \phi(r_1,\dots,r_{k} |\boldsymbol{\alpha}) = \EX_{\alpha_{k}|\alpha_{k+1},\dots,\alpha_n} \{\phi(r_1,\dots,r_{k-1}|\boldsymbol{\alpha})\},\quad  k = 2,\dots, n-1,
\]

where the expectation is taken over $\alpha_{k}$ uniformly distributed on $[n] \setminus \{\alpha_{k+1},\dots,\alpha_n\}$ with  $(\alpha_{k+1},\dots,\alpha_n)$ and $\mathbf{r}$ fixed. This means $d_k$ is the conditional expectation of $e_{k-1}$.

We now rewrite $\phi(\mathbf{r}|\boldsymbol{\alpha})$ as $e_n \prod_{k=2}^{n} e_{k-1}/d_{k}$ since $d_1 = 1$ and start by considering the expectations of the terms  $e_{k-1}/d_{k}, k=2,\dots,n$ with respect to $\alpha_2$ given $\alpha_3,\dots,\alpha_n$.  With the exception of $e_1$ these terms are specified by subsets of $\{\alpha_3,\dots,\alpha_n\}$, so they remain fixed under the conditioning. Since the conditional expectation of $e_1$ is $d_2$, the conditional expectation of $e_n \prod_{k=2}^{n} e_{k-1}/d_{k}$ is $e_n \prod_{k=3}^{n} e_{k-1}/d_{k}$, i.e.\ there is one fewer term in the product. Proceeding in this way and successively taking conditional expectations with respect to $\alpha_{k-1}$ uniformly distributed on $[n] \setminus \{\alpha_{k},\dots,\alpha_n\}$ conditional on $(\alpha_{k},\dots,\alpha_n)$, the product telescopes to $e_n (e_{n-1}/d_n) = p(r_1,\dots,r_n)e_{n-1}/d_n$.  Finally taking expectations over $\alpha_n$, this reduces to $p(r_1,\dots,r_n)$, since $d_n$ is the expected value of $e_{n-1}$. Note a random permutation $(\alpha_1,\dots,\alpha_n)$ is generated starting from $\alpha_n$ in reverse order by successively selecting elements for inclusion at random from what remains. By the chain rule of expectation, we have then shown that $\EX_{\boldsymbol{\alpha}} \{\phi(\newtilde{r}|\boldsymbol{\alpha})\} = p(\newtilde{r})$ as claimed.   
\end{proof}

We can now describe our second and faster \BS/ algorithm.  This algorithm is the basis of our main result, Theorem \ref{thm2n}. We start by sampling a random permutation $\boldsymbol{\alpha} = (\alpha_1, \dots, \alpha_n)$ uniformly distributed on $\pi[n]$. As in Algorithm A we consider a chain of conditional pmfs,  
\[\phi(\newtilde{r}|\boldsymbol{\alpha}) = \prod_{k=1}^n e_k/d_k =\phi(r_1|\boldsymbol{\alpha}) \phi(r_2|r_1,\boldsymbol{\alpha}) \phi(r_3|r_1,r_2,\boldsymbol{\alpha}) \dots \phi(r_n | r_1,r_2,\dots,r_{n-1},\boldsymbol{\alpha}),\]
where now $\phi(r_k|r_1,\dots,r_{k-1},\boldsymbol{\alpha}) = e_k/d_k$, with the notation defined prior to Lemma \ref{thm:phi}.

At stage 1, we sample $r_1$ from the pmf $\phi(r_1|\boldsymbol{\alpha})$. Subsequently for stages $k = 2,\dots,n$, we sample $r_k$ from $\phi(r_k|r_1,\dots,r_{k-1}, \boldsymbol{\alpha})$ with $(r_1,\dots,r_{k-1})$ fixed, exploiting the Laplace expansion for permanents. At completion, the array $(r_1,\dots,r_n)$ at stage $n$ will have been  sampled from the pmf $p(\newtilde{r})$. Sorting $(r_1,\dots,r_n)$ in non-decreasing order, the resulting multiset $\newtilde{z}$ is sampled from the \BS/ distribution $q(\newtilde{z})$.

\begin{algorithm}
  \caption{Boson sampler: single sample $\newtilde{z}$ from q($\newtilde{z})$ in $\bigo{n 2^n + \poly(m,n)}$ time
    \label{alg:better}}
  \begin{algorithmic}[1]
    \Require{$m$ and $n$ positive integers;  $A$ first $n$ columns of $m \times m$ Haar random unitary matrix}
    \State $\newtilde{r} \gets \varnothing$ \Comment{\sc Empty array}
    \State $A \gets \textrm{Permute}(A)$ \Comment{Randomly permute columns of $A$}
    \State $w_i \gets |a_{i,1}|^2, i \in [m]$ \Comment{Make indexed weight array $w$}
    \State $x \gets \textrm{Sample}(w)$ \Comment{Sample index $x$ from $w$}
    \State $\newtilde{r} \gets (\newtilde{r}, x)$ \Comment{Append $x$ to $\newtilde{r}$}
    \For{$k \gets 2 \textrm{ to } n$}
    \State $\Bk \gets A^{[k]}_{\newtilde{r}}$
    \State \sc Compute $\{\perm \Bkl, \ell \in [k] $\} \Comment{As Lemma \ref{lemma1}}
    \State $w_i \gets |\sum_{\ell = 1}^k a_{i,\ell} \perm \Bkl|^2,\; i \in [m]$ \Comment{Using Laplace expansion}
    \State $x \gets \textrm{Sample}(w)$ 
    \State $\newtilde{r} \gets (\newtilde{r}, x)$
    \EndFor
    \Let{$\newtilde{z}$}{\sc IncSort($\newtilde{r}$)} \Comment{Sort $\newtilde{r}$ in non-decreasing order}
    \State \Return{$\newtilde{z}$}
  \end{algorithmic}
\end{algorithm}

\begin{proof}[Proof of Theorem \ref{thm2n}]
It follows from the chain rule for conditional probabilities that for a given fixed $\boldsymbol{\alpha}$ Algorithm B samples from $\phi(\newtilde{r}|\boldsymbol{\alpha})$. As $\boldsymbol{\alpha}$ is uniformly distributed on $\pi[n]$ the algorithm therefore samples from the marginal distribution $\EX_{\boldsymbol{\alpha}} \{\phi(\newtilde{r}|\boldsymbol{\alpha})\}$ which by Lemma \ref{thm:phi} is the pmf $p(\newtilde{r})$ in \eqref{def:pr}.  The complexity is established as follows. At stage 1, the pmf to be sampled is 
\[\phi(r_1|\boldsymbol{\alpha}) = |\perm A_{r_1}^{[n] \setminus \{\alpha_2,\dots,\alpha_n\}}|^2 = |\perm A_{r_1}^{\alpha_1}|^2 = |a_{r_1, \alpha_1}|^2,\quad r_1 \in [m].
\]
Sampling takes $\bigo{m}$ operations. At stage 2 for fixed $r_1$, we need to sample $r_2$ from the pmf $\phi(r_2|r_1,\boldsymbol{\alpha})$.  But since $\phi(r_2|r_1,\boldsymbol{\alpha}) = e_2/d_2$ and $d_2$ does not involve $r_2$, we can sample from the pmf that is proportional to $e_2$ considered   simply as a function of $r_2$ with $r_1$ fixed. Since $e_2$ is proportional to $|\perm A^{\alpha_1,\alpha_2}_{r_1,r_2}|^2$ this involves calculating $m$ permanents of  $2 \times 2$ matrices and then sampling; a further $\bigo{m}$ operations.  

At stage $k$ with $r_1,\dots,r_{k-1}$ already determined, we need to sample $r_k$ from the pmf proportional to 
$|\perm A^{\alpha_1,\dots,\alpha_k}_{r_1,\dots,r_k}|^2$ considered simply as a function of $r_k$.   Let $B = A^{\alpha_1,\dots,\alpha_k}_{r_1,\dots,r_{k}}$ for some arbitrary value of $r_k$, i.e.\ $b_{i,j} = a_{r_i,\alpha_j}$ for $i,j = 1,\dots,k$. Using the Laplace expansion and taking squared moduli
\begin{equation}\label{laplace}
\left|\perm A^{\alpha_1,\dots,\alpha_k}_{r_1,\dots,r_k}\right|^2 = \left|\sum_{\ell = 1}^k a_{r_k,\alpha_{\ell}} \perm \Bkl\right|^2,\; r_k \in [m],
\end{equation}
where $\Bkl$ is defined in the statement of Lemma \ref{lemma1}. From Lemma~\ref{lemma1} all the values $\{\perm \Bkl\}$ can be computed in $\bigo{k2^k}$ time. An array of length $m$ is then obtained from \eqref{laplace} by summing $k$ terms for each $r_k \in [m]$.  Taking a single sample from the pmf proportional to the array takes $\bigo{m}$ time. This gives a total time complexity for stage $k$ of $\bigo{k2^k}$ + $\bigo{mk}$.   

The total time for all stages $k=1,\dots,n$ is therefore $\bigo{n 2^n} + \bigo{m n^2}$.  In practical terms, the exponentially growing term is approximately equivalent to the cost of evaluating two $n \times n$ permanents.  The space usage is dominated by storing a single unnormalised pmf of size $\bigo{m}$ prior to sampling. 
\end{proof}
\paragraph{Remarks:} At the final stage when $r_n$ is selected the value of $\left|\perm A^{\alpha_1,\dots,\alpha_n}_{r_1,\dots,r_n}\right|^2 = 
\left|\perm A_{\newtilde{r}}\right|^2$ will have been computed as in \eqref{laplace}. Hence $q(\newtilde{r})$ the pmf of the sample $\newtilde{r}$  is available at no extra computational cost. Finally note that when only a single value is to be sampled from  $q(\newtilde{r})$ and since the columns of a Haar random unitary matrix are already in a random order, the random permutation $(\alpha_1,\dots,\alpha_n)$ in Algorithm B can be taken to be $(1,2,\dots,n)$ 

\section{Acknowledgements}
RC was funded by EPSRC Fellowship EP/J019283/1. Thank you to Ashley Montanaro, Alex Neville, Anthony Laing and members of the Bristol QET Labs for introducing us to the \BS/ problem and for a number of detailed and helpful conversations. 

\section*{Appendix}
\paragraph{Constant factor overheads in Lemma \ref{lemma1}:}

To compare the constant factor overheads in Lemma \ref{lemma1}, we start by looking at $\perm B$. From Glynn's formula we have
\[
\perm  B = \frac{1}{2^{k-1}}\sum_{\delta}  \left( \prod_{i = 1}^{k} \delta_i \right) \prod_{j \in [k]}  v^*_j(\delta), 
\]
where $\delta  \in \{-1,1\}^{k}$ with $\delta_1 =1$ and  $ v^*_j(\delta) = \sum_{i=1}^{k} \delta_i b_{ij}$. As we work through the values of $\delta$ in Gray code order, updating $v^*_j(\delta), j \in [k]$ requires $k$ additions. Evaluating the product requires $k-1$ multiplications and there is one further addition as the sum is accumulated.  There are $2^{k-1}$ steps making a total of $2 k 2^{k-1} = k 2^k$ operations.  

For the calculation of 
$\{\perm  \Bkl, \ell \in [k]\}$ in Lemma \ref{lemma1}, at each stage there are $k$ additions for updating $\{v_j(\delta), j \in [k]\}$, then $k-2$ multiplications for calculating each of the sets $\{f_\ell\}$,  $\{b_\ell\}$, and $\{f_\ell b_\ell\}$.  Having thereby obtained all the partial products $\prod_{j \in [k]}  v_j(\delta)/v_{\ell}(\delta)$, a further $k$ additions are required as the $k$ partial sums are accumulated. So each of the $2^{k-2}$ steps requires $5 k - 6$ operations. Consequently the operation count is approximately 25\% larger than that of $\perm B$. In practice addition and multiplication with complex arithmetic will have different costs in terms of floating point operations.

\paragraph{Collision probability:}

Suppose that $\newtilde{z}$ is sampled from $q(\newtilde{z})$ in \eqref{eq:Bz_s} and that $\newtilde{s}$ is the associated array of counts.  The probability of duplicate elements in $\newtilde{z}$ is the probability that one of the elements of $\newtilde{s}$ is larger than $1$. In general this probability will depend on the underlying Haar random unitary $A$. 

 From Lemma \ref{th:subsequence} we have  
\begin{align}\label{BSrepeat}
\mathbb{E}(s_i) &= \sum_{k=1}^n |a_{i,k}|^2 \nonumber \\ 
\mathbb{E}(s_i(s_i-1)) &= \sum_{c \in \mathcal{C}_2} \left|\text{Perm }A^c_{i,i}\right|^2 = 
4 \sum_{k<\ell} |a_{i,k} a_{i,\ell}|^2
\end{align}

Let us denote the $A$-specific duplication probability, i.e.\@   the chance of seeing duplicated values in $\newtilde{z}$, by $P_D$.  We have from \eqref{BSrepeat}

\[P_D \leqslant \sum_{i=1}^m \mathbb{P}(s_i \geqslant 2) \leqslant
\frac12 \sum_{i=1}^m \mathbb{E}\,s_i(s_i - 1) = 2 \sum_{i=1}^m
\sum_{k<\ell} |a_{i,k} a_{i,\ell}|^2.
\]

Note that the expected value of the upper bound is $n(n-1)/(m+1) \sim n^2/m$ when averaged over $A$ from the Haar random unitary class. See \citet[Theorem C.4]{AA:2011} and \citet{arkhipov2012bosonic} for comparison.

\bibliographystyle{apalike}
\bibliography{short}

\begin{thebibliography}{}

\bibitem[Aaronson and Arkhipov, 2011]{AA:2011}
Aaronson, S. and Arkhipov, A. (2011).
\newblock The computational complexity of linear optics.
\newblock In {\em STOC '11: Proc. 43\textsuperscript{rd} Annual ACM Symp.
  Theory of Computing}, pages 333--342.

\bibitem[Aaronson and Arkhipov, 2014]{AA:2014}
Aaronson, S. and Arkhipov, A. (2014).
\newblock {Bosonsampling} is far from uniform.
\newblock {\em Quantum Information \& Computation}, 14(15-16):1383--1423.

\bibitem[Arkhipov and Kuperberg, 2012]{arkhipov2012bosonic}
Arkhipov, A. and Kuperberg, G. (2012).
\newblock The bosonic birthday paradox.
\newblock {\em Geometry \& Topology Monographs}, 18:1--7.

\bibitem[Bentivegna et~al., 2015]{bentivegna2015experimental}
Bentivegna, M., Spagnolo, N., Vitelli, C., Flamini, F., Viggianiello, N.,
  Latmiral, L., Mataloni, P., Brod, D.~J., Galv{\~a}o, E.~F., Crespi, A.,
  et~al. (2015).
\newblock Experimental scattershot boson sampling.
\newblock {\em Science advances}, 1(3):e1400255.

\bibitem[Bj{\"o}rklund, 2016]{bjorklund2016below}
Bj{\"o}rklund, A. (2016).
\newblock Below all subsets for some permutational counting problems.
\newblock In {\em 15th Scandinavian Symposium and Workshops on Algorithm Theory
  (SWAT 2016)}, pages 1--17.

\bibitem[Boixo et~al., 2016]{BISBDJMH:2016}
Boixo, S., Isakov, S.~V., Smelyanskiy, V.~N., Babbush, R., Ding, N., Jiang, Z.,
  Bremner, M.~J., Martinis, J.~M., and Neven, H. (2016).
\newblock Characterizing quantum supremacy in near-term devices.
\newblock {\em \href{https://arxiv.org/abs/1608.00263}{arXiv:1608.00263}}.

\bibitem[Broome et~al., 2013]{broome2013photonic}
Broome, M.~A., Fedrizzi, A., Rahimi-Keshari, S., Dove, J., Aaronson, S., Ralph,
  T.~C., and White, A.~G. (2013).
\newblock Photonic boson sampling in a tunable circuit.
\newblock {\em Science}, 339(6121):794--798.

\bibitem[Clifford and Clifford, 2017]{Boson:Rpackage}
Clifford, P. and Clifford, R. (2017).
\newblock {\em R: {C}lassical {B}oson {S}ampling}.
\newblock \url{https://cran.r-project.org/package=BosonSampling}.

\bibitem[Crespi et~al., 2013]{crespi2013integrated}
Crespi, A., Osellame, R., Ramponi, R., Brod, D.~J., Galv{\~a}o, E.~F.,
  Spagnolo, N., Vitelli, C., Maiorino, E., Mataloni, P., and Sciarrino, F.
  (2013).
\newblock Integrated multimode interferometers with arbitrary designs for
  photonic boson sampling.
\newblock {\em Nature Photonics}, 7(7):545--549.

\bibitem[Efraimidis, 2015]{Efraimidis:2015}
Efraimidis, P.~S. (2015).
\newblock Weighted random sampling over data streams.
\newblock In {\em Algorithms, Probability, Networks, and Games}, pages
  183--195. Springer.

\bibitem[Feller, 1968]{Feller:1968}
Feller, W. (1968).
\newblock {\em An Introduction to Probability Theory and its Applications. -
  Vol. 1}.
\newblock Wiley.

\bibitem[Glynn, 2010]{Glynn:2010}
Glynn, D.~G. (2010).
\newblock The permanent of a square matrix.
\newblock {\em European Journal of Combinatorics}, 31(7):1887--1891.

\bibitem[Green et~al., 2015]{GLPR:2015}
Green, P.~J., {\L}atuszy{\'{n}}ski, K., Pereyra, M., and Robert, C.~P. (2015).
\newblock Bayesian computation: a summary of the current state, and samples
  backwards and forwards.
\newblock {\em Statistics and Computing}, 25(4):835--862.

\bibitem[Grover, 1996]{Grover:1996}
Grover, L.~K. (1996).
\newblock A fast quantum mechanical algorithm for database search.
\newblock In {\em STOC '96: Proc. 28\textsuperscript{th} Annual ACM Symp.
  Theory of Computing}, pages 212--219.

\bibitem[Harrow and Montanaro, 2017]{HM:2017}
Harrow, A.~W. and Montanaro, A. (2017).
\newblock Quantum computational supremacy.
\newblock {\em Nature}, 549(7671):203--209.

\bibitem[Hastings, 1970]{hastings1970monte}
Hastings, W.~K. (1970).
\newblock Monte {C}arlo sampling methods using {M}arkov chains and their
  applications.
\newblock {\em Biometrika}, 57(1):97--109.

\bibitem[Kronmal and Peterson~Jr, 1979]{KP:1979}
Kronmal, R.~A. and Peterson~Jr, A.~V. (1979).
\newblock On the alias method for generating random variables from a discrete
  distribution.
\newblock {\em The American Statistician}, 33(4):214--218.

\bibitem[Latmiral et~al., 2016]{LSS:2016}
Latmiral, L., Spagnolo, N., and Sciarrino, F. (2016).
\newblock Towards quantum supremacy with lossy scattershot boson sampling.
\newblock {\em New Journal of Physics}, 18(11).

\bibitem[Liu, 2008]{liu2008monte}
Liu, J.~S. (2008).
\newblock {\em Monte Carlo Strategies in Scientific Computing}.
\newblock Springer Science \& Business Media.

\bibitem[Lund et~al., 2017]{LBR:2017}
Lund, A.~P., Bremner, M.~J., and Ralph, T.~C. (2017).
\newblock Quantum sampling problems, bosonsampling and quantum supremacy.
\newblock {\em NPJ Quantum Information}, 3(1):15.

\bibitem[Marcus and Minc, 1965]{marcus1965permanents}
Marcus, M. and Minc, H. (1965).
\newblock Permanents.
\newblock {\em The American Mathematical Monthly}, 72(6):577--591.

\bibitem[Neville et~al., 2017]{NSCJML:2017}
Neville, A., Sparrow, C., Clifford, R., Johnston, E., Birchall, P.~M.,
  Montanaro, A., and Laing, A. (2017).
\newblock Classical boson sampling algorithms with superior performance to
  near-term experiments.
\newblock {\em Nature Physics}.
\newblock Advance online publication. \doi{10.1038/nphys4270}.

\bibitem[Nijenhuis and Wilf, 1978]{NW:1978}
Nijenhuis, A. and Wilf, H.~S. (1978).
\newblock {\em Combinatorial Algorithms: for Computers and Calculators}.
\newblock Academic press.

\bibitem[Propp and Wilson, 1996]{propp-wilson:exact-sampling}
Propp, J.~G. and Wilson, D.~B. (1996).
\newblock Exact sampling with coupled {M}arkov chains and applications to
  statistical mechanics.
\newblock {\em Random Structures and Algorithms}, 9(1\&2):223--252.

\bibitem[Ryser, 1963]{Ryser:1963}
Ryser, H.~J. (1963).
\newblock {\em Combinatorial Mathematics}, volume~14 of {\em Carus Mathematical
  Monographs}.
\newblock {M}athematical {A}ssociation of {A}merica.

\bibitem[Shchesnovich, 2013]{shchesnovich2013asymptotic}
Shchesnovich, V. (2013).
\newblock Asymptotic evaluation of bosonic probability amplitudes in linear
  unitary networks in the case of large number of bosons.
\newblock {\em International Journal of Quantum Information}, 11(05):1350045.

\bibitem[Shor, 1997]{Shor:1997}
Shor, P.~W. (1997).
\newblock Polynomial-time algorithms for prime factorization and discrete
  logarithms on a quantum computer.
\newblock {\em SIAM Journal on Computing}, 26(5):1484--1509.

\bibitem[Spagnolo et~al., 2014]{spagnolo2014experimental}
Spagnolo, N., Vitelli, C., Bentivegna, M., Brod, D.~J., Crespi, A., Flamini,
  F., Giacomini, S., Milani, G., Ramponi, R., Mataloni, P., et~al. (2014).
\newblock Experimental validation of photonic boson sampling.
\newblock {\em Nature Photonics}, 8(8):615--620.

\bibitem[Spring et~al., 2013]{spring2013boson}
Spring, J.~B., Metcalf, B.~J., Humphreys, P.~C., Kolthammer, W.~S., Jin, X.-M.,
  Barbieri, M., Datta, A., Thomas-Peter, N., Langford, N.~K., Kundys, D.,
  et~al. (2013).
\newblock Boson sampling on a photonic chip.
\newblock {\em Science}, 339(6121):798--801.

\bibitem[Terhal and DiVincenzo, 2004]{TD:2004}
Terhal, B.~M. and DiVincenzo, D.~P. (2004).
\newblock Adaptive quantum computation, constant depth quantum circuits and
  {A}rthur-{M}erlin games.
\newblock {\em Quantum Information \& Computation}, 4(2):134--145.

\bibitem[Tichy, 2011]{Tichy:2011}
Tichy, M.~C. (2011).
\newblock {\em Entanglement and Interference of Identical Particles}.
\newblock PhD thesis, Freiburg University.

\bibitem[Tichy et~al., 2014]{tichy2014stringent}
Tichy, M.~C., Mayer, K., Buchleitner, A., and M{\o}lmer, K. (2014).
\newblock Stringent and efficient assessment of boson-sampling devices.
\newblock {\em Physical Review Letters}, 113(2):020502.

\bibitem[Tillmann et~al., 2013]{tillmann2013experimental}
Tillmann, M., Daki{\'c}, B., Heilmann, R., Nolte, S., Szameit, A., and Walther,
  P. (2013).
\newblock Experimental boson sampling.
\newblock {\em Nature Photonics}, 7(7):540--544.

\bibitem[Valiant, 1979]{Valiant:1979}
Valiant, L. (1979).
\newblock The complexity of computing the permanent.
\newblock {\em Theoretical Computer Science}, 8(2):189 -- 201.

\bibitem[Walker, 1974]{Walker:1974}
Walker, A.~J. (1974).
\newblock New fast method for generating discrete random numbers with arbitrary
  frequency distributions.
\newblock {\em Electronics Letters}, 10(8):127--128.

\bibitem[Walschaers et~al., 2016]{walschaers2016statistical}
Walschaers, M., Kuipers, J., Urbina, J.-D., Mayer, K., Tichy, M.~C., Richter,
  K., and Buchleitner, A. (2016).
\newblock Statistical benchmark for bosonsampling.
\newblock {\em New Journal of Physics}, 18(3):032001.

\bibitem[Wang et~al., 2016]{Wang2016}
Wang, H., He, Y., Li, Y.-H., Su, Z.-E., Li, B., Huang, H.-L., Ding, X., Chen,
  M.-C., Liu, C., Qin, J., Li, J.-P., He, Y.-M., Schneider, C., Kamp, M., Peng,
  C.-Z., Hoefling, S., Lu, C.-Y., and Pan, J.-W. (2016).
\newblock Multi-photon boson-sampling machines beating early classical
  computers.
\newblock {\em \href{https://arxiv.org/abs/1612.06956}{arXiv:1612.06956}}.

\bibitem[Wang and Duan, 2016]{wang2016certification}
Wang, S.-T. and Duan, L. (2016).
\newblock Certification of boson sampling devices with coarse-grained
  measurements.
\newblock {\em Bulletin of the American Physical Society}.

\bibitem[Wu et~al., 2016]{wu2016computing}
Wu, J., Liu, Y., Zhang, B., Jin, X., Wang, Y., Wang, H., and Yang, X. (2016).
\newblock {Computing permanents for boson sampling on Tianhe-2 supercomputer}.
\newblock {\em \href{https://arxiv.org/abs/1606.05836}{arXiv:1606.05836}}.

\end{thebibliography}
\end{document}